\documentclass{article}

\usepackage{amsfonts}
\usepackage{amssymb}
\usepackage{amsthm}
\usepackage{graphicx}
\usepackage{subfigure}

\newcommand{\e}{\varepsilon}
\newcommand{\s}{\sigma}
\renewcommand{\l}{\lambda}
\renewcommand{\O}{\Omega}
\newcommand{\g}{\gamma}

\theoremstyle{plain}
\newtheorem{thm}{Theorem}
\newtheorem{lem}{Lemma}

\theoremstyle{definition}
\newtheorem{definition}{Definition}

\title{Orphan-Free Anisotropic Voronoi Diagrams}
\author{Guillermo D. Canas and Steven J. Gortler\\ School of Engineering and Applied Sciences\\ Harvard University}
\date{}

\begin{document}
\maketitle

\begin{abstract}

We describe conditions under which an appropriately-defined anisotropic
Voronoi diagram of a set of sites in  Euclidean space 
is guaranteed to be composed of connected cells 
in any number of dimensions. 
These conditions are natural for problems in optimization and approximation, 
and algorithms already exist to produce sets of sites that satisfy them. 
\end{abstract}


\section{Introduction}

The anisotropic Voronoi diagram (AVD) is a fundamental data structure with wide practical application. 
In the definition of~\cite{LL2000}, an AVD over a Riemannian manifold is the Voronoi diagram of a set of sites with respect to 
  the geodesic distance associated with a Riemannian metric. 
From a formal viewpoint, this definition has several strengths. 
For instance, a simple argument can be used to show that the Voronoi regions of such an AVD 
are always connected. 

For problems in optimization, such as vector quantization on a Riemannian manifold, the above property makes it simpler to compute the AVD: 
	since each Voronoi region is connected we do not have to search for possible disconnected (orphan) pieces elsewhere. 
For problems in approximation, where we are often more interested in the dual simplicial complex, 
	the absence of orphan regions in the AVD makes it possible for strict approximation error bounds to be enforced on the dual and, 
	in some cases, it can make it possible to characterize the asymptotic size of the approximation~\cite{Gruber,enets}.
Geodesic distance on a Riemannian manifold is, however, very expensive to compute, since it involves finding the shortest geodesic path between two points, 
	and a practical AVD algorithm based on this distance is not currently feasible. 

By choosing a particular parametrization of a Riemannian manifold over a subset of Euclidean space, 
two approximations to the geodesic distance between two points naturally arise, both of which have been considered as basis for constructing AVDs. 
These approximations simplify the problem by considering the metric constant along any path between the two points, 
	but use a different choice of constant metric along the path. 
	In particular, to measure the distance between a given site and any point in the domain, 
		one approximation evaluates the metric only at the site, 
		while the other uses the metric at the point. 
	They are described in~\cite{LS} and~\cite{DW}. 
	We refer to them as Labelle/Shewchuk, and Du/Wang diagrams, respectively, or LS and DW diagrams for short. 
Although they are conceptually similar, guarantees of well-behaved-ness 
	have only been shown only for the case of LS diagrams, and only in two dimensions~\cite{LS}. 
These guarantees 
	come in the form of a condition which, if satisfied, guarantees connectedness of Voronoi cells 
	as well as the well-behavedness of its dual 2D triangulation (absence of triangle inversions).
	They also provide an iterative site-insertion algorithm that is guaranteed to produce a well-behaved output 
	if run for long enough. 

In this paper we present conditions under which these anisotropic 
	Voronoi diagrams 
	are well-behaved. 
These conditions are simple and intuitive, hold for any number of dimensions, and apply to both LS and DW diagrams. 

The condition we describe 
	requires that the generating set of sites form a sufficiently dense $\e$-net (with $\e$ being sufficiently small). 
This condition is quite natural, since it requires that the distribution of generating sites be roughly ``uniform" with respect to the input metric. 
	An iterative, greedy site-insertion algorithm~\cite{Gonz} exists to compute sets that satisfy it, 
		which works by iteratively inserting a new site at the point in the domain that is furthest from the current set of sites. 
Note that the net condition is natural both for optimization problems, such as quantization or clustering, 
	where we are interested in the primal Voronoi diagram~\cite{Gruber,VSA,KM++},
	as well as for approximation problems, where we are mostly interested in the dual simplicial complex. 
	In the former case, optimal quantization sites have been shown to satisfy the net property~\cite{Gruber,enets}, 
		while in the latter, nets have been used to construct asymptotically-optimal approximations~\cite{enets}.

The net condition is both a condition on the density of sites (cover property), and its relative distribution (packing property). 
While the packing property is clearly insufficient to guarantee that an AVD is well-behaved, 
	we show that the density property alone can be insufficient in cases when the net condition is. 
That is, even though there may be sufficiently dense covers that produce well-behaved AVDs, 
	there are $\epsilon$-nets that produce well-behaved AVDs where the $\epsilon$-cover property alone would not (the converse is not possible since every net is a cover). 
For sufficiently fine densities of sites, the combination of cover (density), and packing (relative position) properties is enough to ensure that both the DW and LS diagrams are well-behaved, 
	in any number of dimensions. 
Bounds on the minimum density required are given which, perhaps surprisingly, do not depend on the dimension. 


\section{Preliminary definitions}\label{Vsec:defs}


Given an Euclidean domain $\O\subset\mathbb{R}^n$, an ordinary Voronoi diagram of a set $V\subset\O$ is a partition of the domain into regions whose points are closest to the same element in $V$. 
In the case of LS or DW diagrams, the function used to measure closeness is not the natural distance. 
If we assume that we are given a Riemannian metric over $\O$, with coordinates $Q:\O\rightarrow\mathbb{R}^{n\times n}$ (where at each $p\in\O$, $Q_p$ is symmetric, positive definite), 
	we define the functions: $D_Q^{{ }^{LS}}(a,b) = \left[(a-b)^t Q_a (a-b)\right]^{1/2}$, and $D_Q^{{ }^{DW}}(a,b) = \left[(a-b)^t Q_b (a-b)\right]^{1/2}$
	(note $D_Q^{{ }^{LS}}(a,b) = D_Q^{{ }^{DW}}(b,a)$). 
An LS diagram (resp.\ DW diagram) of a set $V$ is the Voronoi diagram of $V$ with respect to the function $D_Q^{{ }^{LS}}$ (resp.\ $D_Q^{{ }^{DW}}$). 
The associated Voronoi regions of a site $v\in V$ are, respectively, 
	\[ R^{{ }^{LS}}_v = \{ p\in\O | \forall w\in V, D_Q^{{ }^{LS}}(v,p) \le D_Q^{{ }^{LS}}(w,p) \} \]
and
	\[ R^{{ }^{DW}}_v = \{ p\in\O | \forall w\in V, D_Q^{{ }^{DW}}(v,p) \le D_Q^{{ }^{DW}}(w,p) \} \]
Note that, because neither $D_Q^{{ }^{DW}}$ nor $D_Q^{{ }^{LS}}$ are symmetric, we must follow some convention on the order of its arguments. 
In particular, we place an element of $V$ always as first argument of $D_Q^{{ }^{LS}}$ and $D_Q^{{ }^{DW}}$. 
We follow this argument order convention for the rest of the paper. 

%
\vspace*{0.15in}
\noindent
{\bf Asymmetric $\e$-net}. 
In the remainder of the paper, it will be useful to consider $\e$-nets 
	with respect to a function that is not symmetric (e.g. $D_Q^{{ }^{DW}}$ and $D_Q^{{ }^{LS}}$). 
	Since the original definition of an $\e$-net assumes the use of a symmetric distance, 
	we must slightly modify it 
	in this asymmetric case.
As before, we follow the convention of placing elements of the net always 
	as first arguments to $D_Q^{{ }^{DW}}$ or $D_Q^{{ }^{LS}}$. 

\begin{definition}
\emph{An asymmetric $\e$-net with respect to a 
			function $D:\O\times\O\rightarrow\mathbb{R}$ is a set $V\in\O$ that satisfies:
	\begin{enumerate}
		\item $\forall p\in\O$, $\displaystyle{\min_{v\in V} D(v,p)} \le \e$. (asymmetric $\e$-cover property)
		\item $\forall v,w\in V$, $D(v,w) > \e$ \emph{or} $D(w,v) > \e$. (asymmetric $\e$-packing property)
	\end{enumerate}
}
\end{definition}
These properties are analogous to those of a regular net, but not identical (e.g. the packing property is weaker). 
Note that, even if $D$ above is not symmetric, we can still compute an asymmetric net using 
the greedy algorithm of~\cite{Gonz}, by being careful to follow the stated argument-order convention.
A simple induction argument reveals that, in the asymmetric case, the algorithm of~\cite{Gonz} always terminates by outputting a discrete set $V$ that satisfies the properties of an asymmetric net.
(In particular, the output of~\cite{Gonz} will not, in general, satisfy the stronger version of the asymmetric packing property: $\forall v,w\in V$, $D(v,w) > \e$ \emph{and} $D(w,v) > \e$). 

\section{Setup}\label{Vsec:setup}

Assume that we are given as input a 
metric over n-dimensional Euclidean space (in coordinates: a field of symmetric, positive definite (PD) matrices $Q:\O\rightarrow\mathbb{R}^{n\times n}$).
At every point in $\O$, there is an eigen-decomposition $Q=R\Lambda R^t$. 
	Let the symmetric PD square root matrix of $Q$ be $M = R\Lambda^{1/2}R^t$. 
As in~\cite{LS}, we use this square root matrix to analyze the Voronoi diagrams under consideration.  
In contrast to~\cite{LS}, we do not consider {any} square-root matrix $M'$ such that $M'^t M'=Q$ (all square roots can be written as $U M$, where $U$ is unitary: $U^t U=I$), but rather 
	concentrate on the unique square root that is also symmetric and positive definite.  
This distinction is important. In particular, as the following lemma shows, the symmetric PD square root has the same differentiability class as $Q$, 
	and, in particular, it is continuous wherever $Q$ is. \\

\begin{lem}
\label{Vlm:sqrtQ}
	Given $Q:\O\subset\mathbb{R}^n\rightarrow\mathbb{R}^{n\times n}$, symmetric, positive definite over a compact domain $\O$, 
		then the unique symmetric, positive definite matrix $M$ that satisfies $M^t M = Q$ is $M\in\mathcal{C}^k$ if and only if $Q\in\mathcal{C}^k$ 
\end{lem}
\begin{proof}
	If $M\in\mathcal{C}^k$ then clearly $Q=M^t M\in\mathcal{C}^k$. 
	To show that the converse is also true assume $Q\in\mathcal{C}^k$ and consider,  
	at any point $p\in\O$, the eigen-decomposition $Q=R\Lambda R^t$, where $\Lambda$ is a diagonal matrix with positive entries.
	We claim that $M=R\Lambda^{1/2} R^t$. 
	Clearly, $M$ is symmetric, positive definite, and $M^t M=Q$.
	Any other positive definite matrix $M'$ with $M'^t M'=Q$ is related to $M$ by a unitary matrix as $M'=U M$. 
		But, since $M$ is symmetric, $M'$ can only be symmetric if $U=I$. 
		Therefore $M=R\Lambda^{1/2} R^t$ is the only square root of $Q$ that is symmetric PD. 
	
	Just like the Taylor expansion of the (scalar) square root function $\sqrt{x}=\sum_{n=0}^{\infty}\frac{(-1)^n(2n)!}{(1-2n)(n!)^2(4^n)}(x-1)^n$, which converges for $|x-1|<1$, 
		we can define the series $S=\sum_{i=0}^{\infty}\frac{(-1)^i(2i)!}{(1-2i)(i!)^2(4^i)}(Q-I)^i$, 
		which, by d'Alembert's ratio test~\cite{Rudin}, is absolutely convergent for symmetric PD matrices with spectral radius $\rho(Q)<2$ 
				since the absolute value of the coefficients in the series is decreasing, and the norm of $(Q-I)^i$ is exponentially decreasing with $i$. 
	We show that if $\rho(Q)< 2$, then the series converges to $M$, the unique symmetric PD square root. 
	Clearly, $S$ is symmetric and, given an eigenvector $e_i$ of $Q$ with eigenvalue $\l_i$, 
		it is $e_i^t S e^i = \sum_{n=0}^{\infty}\frac{(-1)^n(2n)!}{(1-2n)(n!)^2(4^n)}(\l_i -1)^n=\l_i^{1/2}$.
	Hence $S=M$. 
	
	The derivatives $S_n^{(j)} = \sum_{i=0}^{n}\frac{(-1)^i(2i)!}{(1-2i)(i!)^2(4^i)} d^{(j)}(Q-I)^i$ of the partial sums are defined for $j=0,\dots,k$, and have the same convergence region as the original series $S$. 
	Since $Q$ is symmetric PD, $\rho(Q)<2$ and $\O$ is compact, $Q-I$ has maximum spectral norm $\rho_{\max}=\max_{p\in\Omega} \rho(Q-I)<1$. 
		Thus the slowest rate of convergence occurs at points $p\in\Omega$ where $\rho(p)=\rho_{\max}$, and therefore the $S_n^{(j)}$, $j=0,\dots,k$ converge uniformly inside $\O$. 
		Uniform convergence implies that the limit of derivatives is the derivative of the limit~\cite{Rudin}: $\lim_{n\rightarrow\infty} S_n^{(j)} = d^{(j)}S = d^{(j)} M$. 
		Therefore $M$ is at least in the same differentiability class as $Q$. 
	
	The above applies to matrices $Q$ with $\rho(Q)<2$. In the general case, scale $Q$ by $1/\max_{p\in\O}\rho(Q_p)$, 
		apply the lemma, and rescale back the resulting $M$ by $\max_{p\in\O}\rho(Q_p)^{1/2}$. 
\end{proof}
\vspace*{0.2in}

Because $Q$ is spatially-varying, given two points $a,b\in\O$, it will in general be $D_Q^{{ }^{DW}}(a,b) \neq D_Q^{{ }^{DW}}(b,a)$, and likewise for $D_Q^{{ }^{LS}}$. 
The amount of asymmetry will depend on how different $Q_a$ is from $Q_b$. 
In particular, by the sub-multiplicative property of the spectral matrix norm, 
\begin{eqnarray*}
	{D_Q^{{ }^{DW}}(a,b)}/{D_Q^{{ }^{DW}}(b,a)} &=& {D_Q^{{ }^{LS}}(b,a)}/{D_Q^{{ }^{LS}}(a,b)} = {\|M_b (a-b)\|}/{\|M_a (a-b)\|} \\
			&=& \frac{\| M_b M_a^{-1} M_a (a-b) \| }{\| M_a (a-b) \|} \le \rho(M_b M_a^{-1})
\end{eqnarray*}
and a similar argument shows that $D_Q^{{ }^{DW}}(a,b) / D_Q^{{ }^{DW}}(b,a) \ge \rho_m(M_b M_a^{-1})$, 
	where $\rho(A)$ is the spectral norm 
	and $\rho_m(A)$ is the smallest of the absolute values of the eigenvalues of $A$
	(that is: $\rho(A) ={\sup_{r\ne 0} \| A r \|/ \|r\|}$ and $\rho_m(A) ={\inf_{r\ne 0} \|A r\| / \|r\|}$). 
Note that, for square matrices, such as the ones considered here, the spectral norm is the same as the induced $L^2$ operator norm.

Given a point $a\in\O$, it is possible to bound the amount of asymmetry in $D_Q^{{ }^{DW}}$ (resp. $D_Q^{{ }^{LS}}$) inside an appropriately-defined neighborhood of $a$. 
To this end, we introduce the following definition, which is applicable to continuous metrics, whose square root $M$ is, by Lemma~\ref{Vlm:sqrtQ}, also continuous. 

\begin{definition}\label{def:sigma0}
The \emph{maximum variation}  of a $\mathcal{C}^0$ metric $Q$ is the smallest constant $\s$ such that for all $a,b\in\O$, it is
\[   \rho\left(M_b M_a^{-1} - I\right) \le \s \cdot \| M_a (a-b) \|   \]
where $M$ is the symmetric, positive definite square root of $Q$. 
\end{definition}
\vspace*{0.2in}
Loosely speaking, $\s$ is a Lipschitz-type bound on the rate of variation of $M$ relative to itself.
In the sequel, it will be assumed that $\s$ is finite. 
In particular, this will always be the case if $\O$ is compact.

We can use the above definitions to find bounds on 
the asymmetry in the associated function $D_Q^{{ }^{DW}}$ (resp. $D_Q^{{ }^{LS}}$). 
The following lemma 
	shows that, if a point $b\in\O$ is inside a certain neighborhood of point $a\in\O$, and $\s$ is sufficiently small, 
	then $\rho(M_b M_a^{-1})$ can be bounded from above, and $\rho_m(M_b M_a^{-1})$  from below, 
	which implies that $D_Q^{{ }^{DW}}(a,b)$ and $D_Q^{{ }^{DW}}(b,a)$ (resp. $D_Q^{{ }^{LS}}(a,b)$ and $D_Q^{{ }^{LS}}(b,a)$) must be similar. 
	
\begin{lem}\label{Vlm:cover0}
	Given $\e>0$ and $Q\in\mathcal{C}^0$ with maximum variation $\s$, then for all $a,b\in\O$ with $\|M_a(a-b)\|\le\e$ it is 
		\[1-\e\s \le  \rho_m(M_b M_a^{-1}) \le {\|M_b (a-b)\|}/{\|M_a (a-b)\|}  \le \rho(M_b M_a^{-1}) \le 1+\e\s \]
\end{lem}
\begin{proof}
The upper bound follows from the sub-multiplicative property of the spectral norm and the definition of $\s$:
	\begin{eqnarray*}
	  {\| M_b (a-b) \|}/{\| M_a (a-b)\|} &=& {\| M_b M_a^{-1} M_a (a-b)\|}/{\| M_a (a-b)\|} \le  \rho(M_b M_a^{-1}) \\
	  			&=& \rho(M_b M_a^{-1} - I + I)  \le 1 + \rho(M_b M_a^{-1} - I) \\
				&\le& 1 + \|M_a(a-b)\|\s \le 1 + \e\s
	\end{eqnarray*}

If $\l_i$ are the eigenvalues of a symmetric matrix $A$, then 
	\begin{eqnarray*}
		\rho_m(A+I) &=& \min_i |\l_i + 1| \ge \min_i | |1| - |\l_i | | \\
					&\ge& \min_i 1 - |\l_i| =  1 - \max_i | \l_i | = 1 - \rho(A)
	\end{eqnarray*}
	and therefore
	\begin{eqnarray*}
	  {\| M_b (a-b) \|}/{\| M_a (a-b)\|} &= & {\| M_b M_a^{-1} M_a (a-b)\|}/{\| M_a (a-b)\|} \\
	  		&\ge& \rho_m(M_b M_a^{-1} - I + I)  \ge 1 - \rho(M_b M_a^{-1} - I) \\
			&\ge& 1 - \|M_a(a-b)\|\s \ge 1 - \e\s
	\end{eqnarray*}
\end{proof}
\vspace*{0.1in}

In the remainder of the paper, we establish conditions for anisotropic Voronoi diagrams to be orphan-free. 
In particular, the set of sites considered will be asymmetric $\e$-nets, where $\e$ must be sufficiently small in relation to the constant $\s$ above, 
	which depends on the input metric and, vaguely speaking, provides an upper bound on the rate of change of the metric.

Note, however, that, in practice, it is not necessary to compute $\s$ to find a sufficiently small $\e$ that guarantees that Voronoi regions are well-behaved. 
Instead, it is possible to simply run the greedy algorithm of~\cite{Gonz}, which in our case outputs asymmetric $\e$-nets whose $\e$ decreases with each iteration, 
	until the resulting Voronoi diagram is orphan free. This is because, at each iteration, the algorithm of~\cite{Gonz} must compute the closest site to each point in the domain -- 
		a task that is equivalent to computing the Voronoi diagram for the current set of sites. Therefore checking at each stage whether the current diagram is orphan-free can simply be a by-product of the 
		asymmetric $\e$-net computation algorithm. 
The proofs in this paper simply guarantee that there is a small-enough $\e$ for which the resulting asymmetric $\e$-net produces an orphan-free Voronoi diagram, 
	and thus that the above algorithm stops (a proof involves the fact that the Voronoi radius of a site can be made as small as desired by simply introducing more sites in the $\e$-net, 
		which follows from the definition of the DW/LS distances). 
The precise bounds in Theorems~\ref{Vth:noorphans0} and~\ref{Vth:LSnoorphans0} may also serve to give some indication of how small $\e$ will need to be.

\section{Orphan-free anisotropic Voronoi diagrams}\label{partI}

This section shows that, given a continuous metric, the associated DW and LS diagrams of an asymmetric $\e$-net, for sufficiently small $\e$, are orphan-free. 
Theorems~\ref{Vth:noorphans0} and~\ref{Vth:LSnoorphans0} 
	state conditions under which this holds. 
In particular the criteria for $\e$ to be sufficiently small will be a certain relation between $\e$ and $\s$, 
	the maximum variation of the metric, which 
	does not depend on the dimension. 
	(Specifically, 
		$\e\s\le 0.09868$ for DW diagrams, and $\e\s\le 0.0584$ for LS diagrams.)
	
Since the (asymmetric) $\e$-net property 
	is a combination of an $\e$-packing and $\e$-cover properties, 
	it is natural to consider whether any of these properties is, by itself, sufficient to guarantee that DW and LS diagrams are well-behaved. 
Clearly, the $\e$-packing property cannot be sufficient since any two sites form an $\e$-packing for any $\e$ smaller than the distance between them. 
The case of $\e$-covers is more subtle: for a particular choice of metric $Q$, 
	it is possible that a sufficiently small $\e'$ exists such that every $\e'$-cover produces well-behaved DW and LS diagrams. 
When we consider all possible choices of $Q$, however, the required $\e'$ may be arbitrarily small relative to $\s$ 
and, unlike for $\e$-nets, there is no constant $c>0$ such that, for all choices of $Q$, an $\e$-cover with $\e\s\le c$ is always guaranteed to produce orphan-free DW and LS diagrams.

\begin{figure}[t]
\begin{center}
\subfigure[]{\includegraphics[width=2.0in]{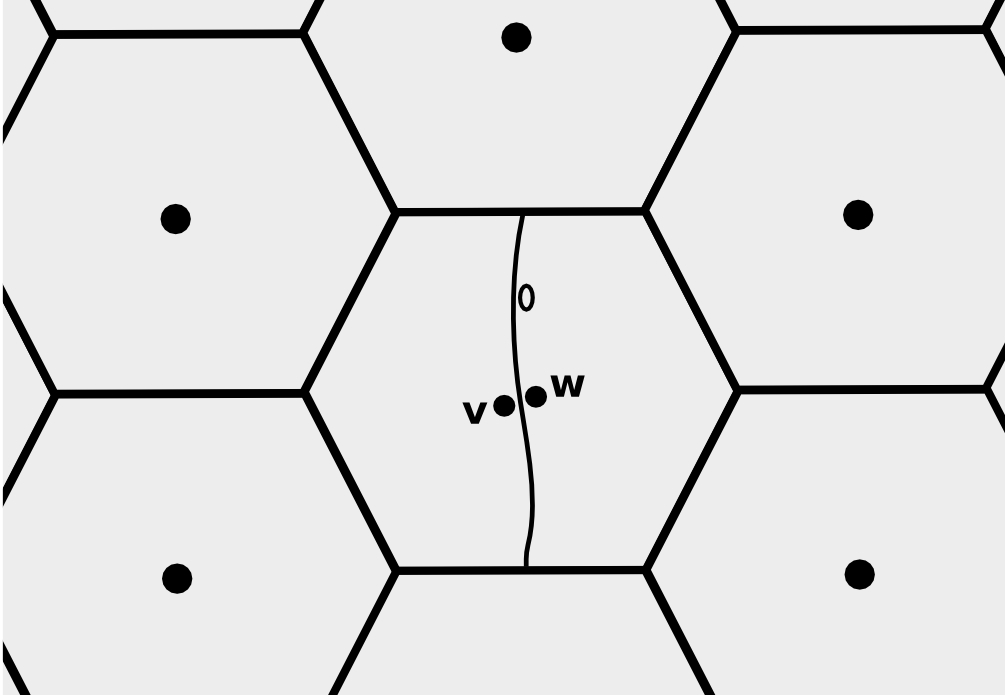}}\quad\quad
\subfigure[]{\includegraphics[width=2.0in]{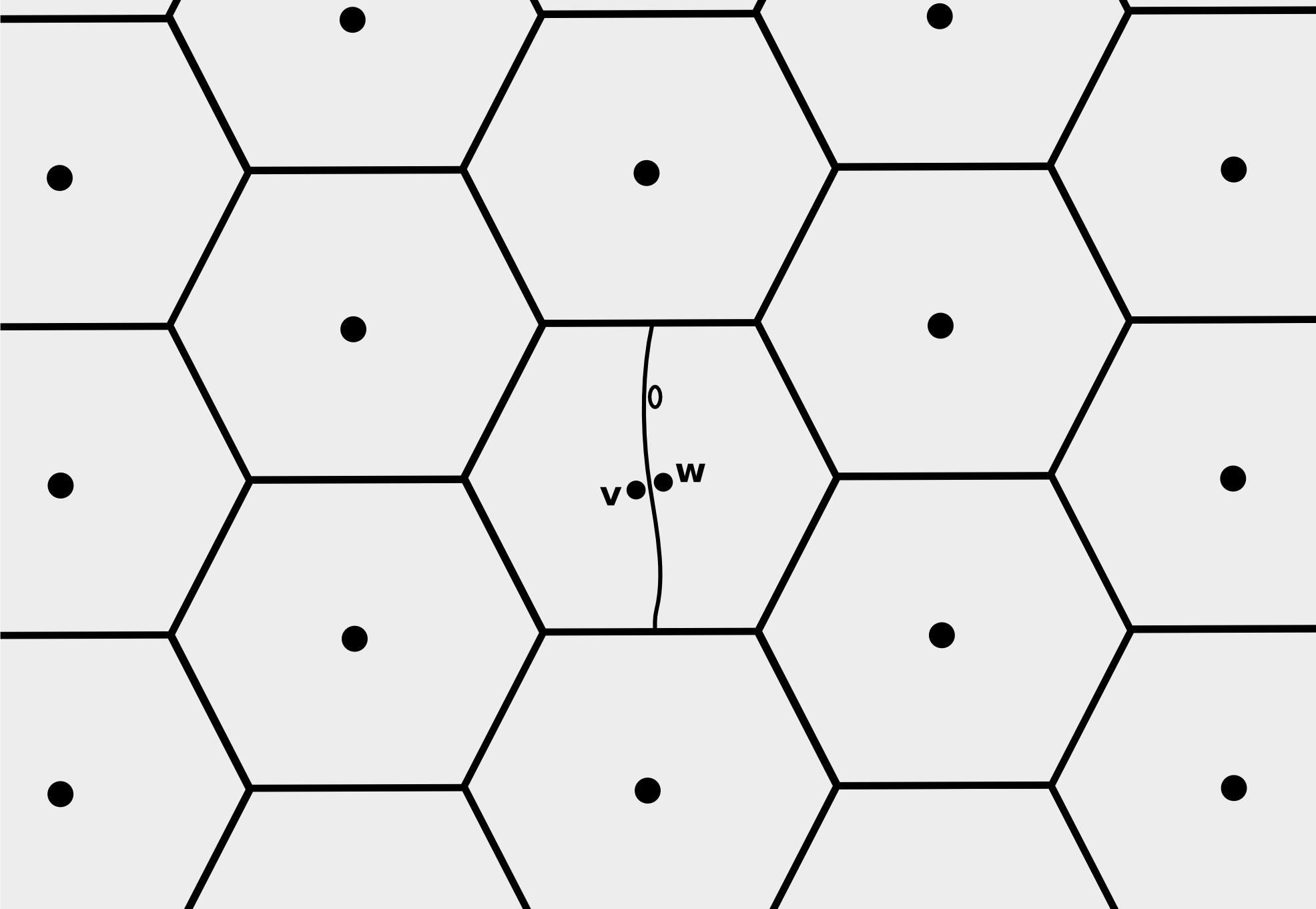}}
\subfigure[]{\includegraphics[width=2.0in]{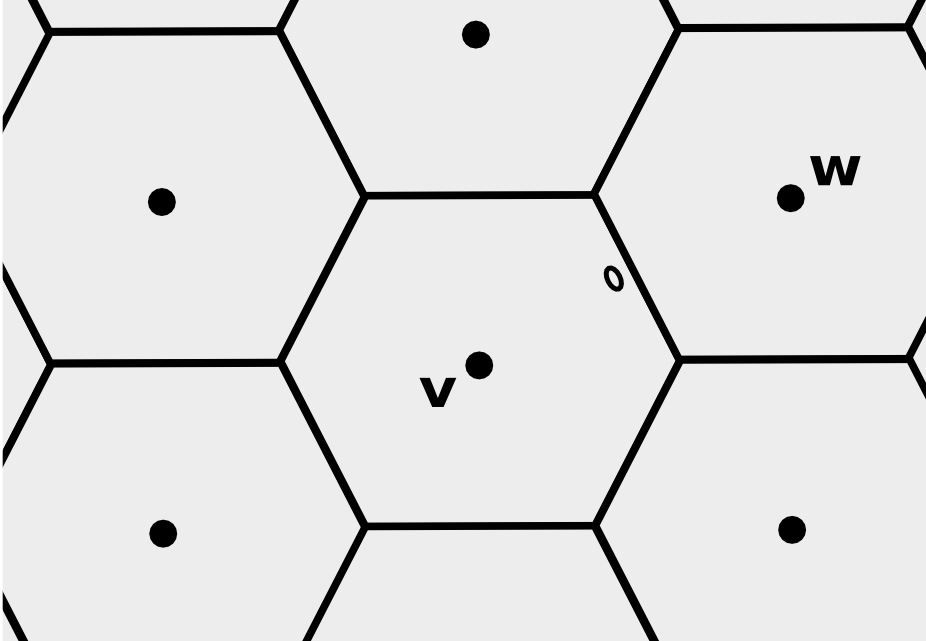}}\quad\quad
\subfigure[]{\includegraphics[width=2.0in]{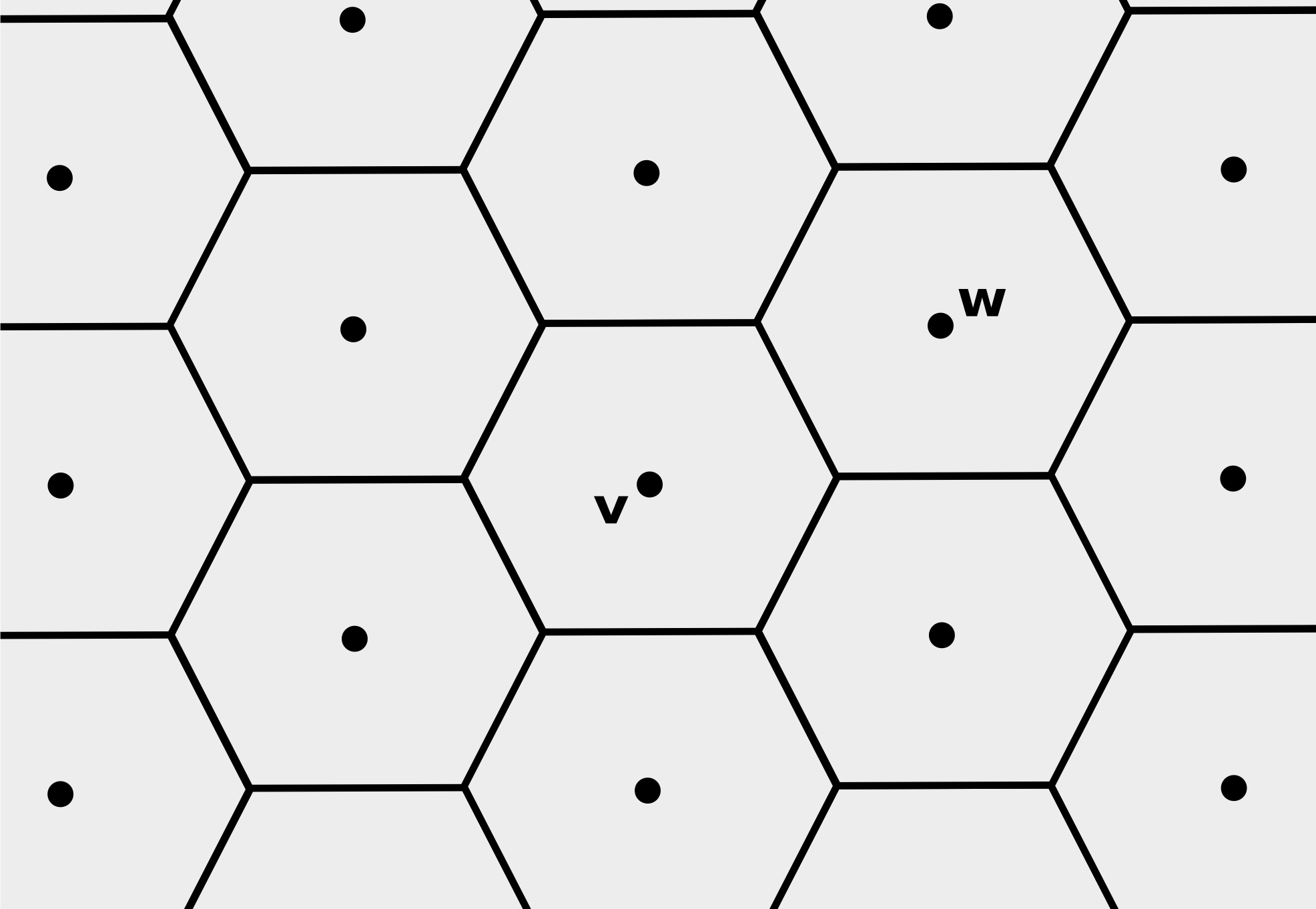}}
\caption{
If two sites ($v$ and $w$) are very close-together, a small perturbation in $D_Q^{{ }^{DW}}$ can cause an orphan (top left), even if the sites are very dense with respect to $D_Q^{{ }^{DW}}$ (top right). 
If sites are sufficiently spaced apart relative to each other (bottom left), orphan regions need sufficiently-large (relative) fluctuations in $D_Q^{{ }^{DW}}$ to appear. 
	In this case, placing the sites more densely-packed with respect to $D_Q^{{ }^{DW}}$ reduces the relative fluctuation, which eliminates the orphans (bottom right). 
}
\label{fig:cover}
\end{center}
\end{figure}

To see this, consider the diagram of Fig.~\ref{fig:cover}. 
In the right column, the set of sites (black dots) is dense enough over $\O$ so that $Q$ is roughly constant inside each Voronoi region. 
There is, however, some small variation in $Q$. 
We can place two sites $v,w\in V$ very close together (but not coinciding) such that, even for small $\e$, a very small change in $Q$ away from them causes an orphan region to appear (top right). 
A point in the orphan region, near the interface between $v$'s and $w$'s Voronoi regions, 
	``sees" both $v$ and $w$ as being at approximately the same distance, 
	but a very small variation in $Q$ has made the points in the orphan region be slightly closer to $v$, even though they are surrounded by points that are slightly closer to $w$. 
Although for a particular choice of metric $Q$ there may be a sufficiently small choice of $\e$ for which all $\e$-covers are guaranteed to produce well-behaved DW and LS diagrams, 
	the fact that the variation in $Q$ described above can be arbitrarily small means that the requirement on $\e$ may be arbitrarily strict, depending on the choice of $Q$. 

\subsection{Orphan-free Du/Wang diagrams}\label{Vsec:noorphans0}

\begin{figure}[t]
\begin{center}
\subfigure[]{\includegraphics[width=2.0in]{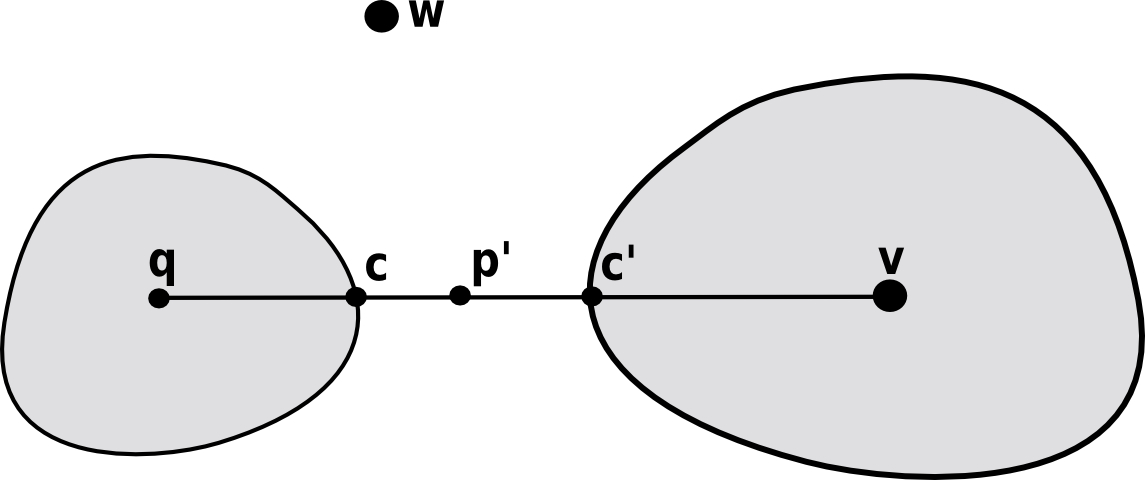}}\quad\quad
\subfigure[]{\includegraphics[width=2.0in]{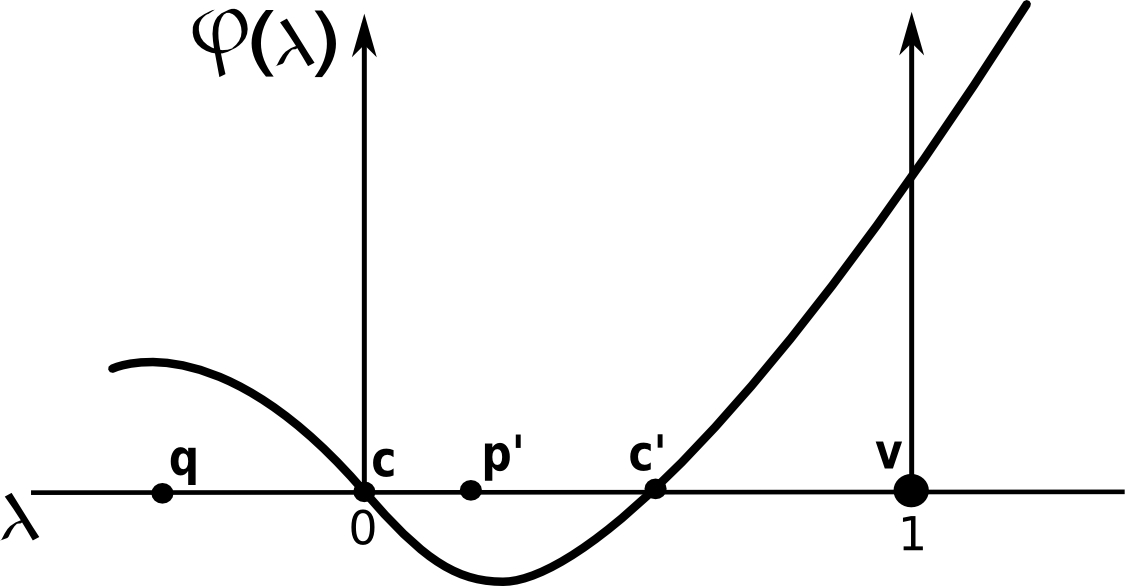}}\quad\quad
\caption{The diagrams for the proofs of Sections~\ref{Vsec:noorphans0}, and~\ref{Vsec:LSnoorphans0}.}
\label{Vfig:noorphans0}
\end{center}
\end{figure}

Given a continuous metric $Q$ and its associated distance $D_Q^{{ }^{DW}}$, 
	we show here that an asymmetric $\e$-net with respect to $D_Q^{{ }^{DW}}$, with sufficiently small $\e$, has Voronoi regions that are 
	always star-shaped (with respect to the Euclidean distance) from their generating sites
	and, in particular, they are connected. 
We prove this by showing that, if a point $q\in\O$ belongs to an orphan region of some site $v\in V$, as in the diagram (Fig.~\ref{Vfig:noorphans0}.a), 
	then the segment $\overline{qv}$ connecting them must also belong to the Voronoi region of $v$, 
	contradicting the fact that the Voronoi region that contains $q$ is disconnected from $v$. 

More specifically, 
since $q$ is in an orphan region, 
the segment $\overline{qv}$ must contain a point $p'$ 
	that is closer to a different site $w\in V$ ($p'$ is in the interior of the Voronoi region of $w\ne v$). 
However, we show that $p'$ cannot be closer to any $w\ne v$, reaching a contradiction. 
In conclusion, $\overline{qv}$ must belong to the Voronoi region of $v$, and so $q$ cannot be in an orphan region. 
Additionally, this shows that every Voronoi region must be star-shaped with respect to its generating site. 

The details of the proof follow. 
In particular, we will see that, if $q$ is in an orphan region of $v\in V$, and $Q$ is continuous, then by the intermediate value theorem, there must be 
	two points $c,c'\in\overline{qv}$ that are at equal distance to $v$ and to some other $w\in V$. 
	We use the $\e$-cover and $\e$-packing properties of $V$ to show that, for sufficiently small $\e$, the existence of such $c,c'$ leads to a contradiction. 

Assume 
	$q\in R^{{ }^{DW}}_v$ and $p'\in R^{{ }^{DW}}_w$, where $R^{{ }^{DW}}_v$, $R^{{ }^{DW}}_w$ are the Voronoi regions of $v$ and $w$ respectively. 
Then there must be a point $c\in\overline{qp'}$ between them that belongs to both Voronoi regions, such that $D_Q^{{ }^{DW}}(v,c) = D_Q^{{ }^{DW}}(w,c)$. 
	Since $V$ is an $\e$-cover, this point must also satisfy $D_Q^{{ }^{DW}}(v,c) = D_Q^{{ }^{DW}}(w,c) \le \e$. 

Consider now the parametrized segment $p(\l) = c (1-\l) + v \l$ with $\l\in[0,1]$. 
Letting $m=(v+w)/2$, define the function
%
\[ 
	\phi(p(\l)) = \e^{-2} \frac{1}{2} \left[D_Q^{{ }^{DW}}(w,p(\l)) - D_Q^{{ }^{DW}}(v,p(\l))\right] = \e^{-2} (w-v)^t Q_{p(\l)} (m - p(\l))
\]
which is plotted in Fig.~\ref{Vfig:noorphans0}.b. Note that $\phi(p(\l))$ is continuous with respect to $\l$ by virtue of the fact that, by Lemma~\ref{Vlm:sqrtQ}, it is $M\in\mathcal{C}^0$.

Since $c$ is equidistant to $v,w$, it is $\phi(c)=0$. 
For $\l>0$, $\phi$ becomes negative at $p'$, since $p'$ is closer to $w$ than to $v$, and then 
	becomes positive at $v$ (since $v$ is closer to $v$ than to $w$). 
Because it is $\phi(p')<0$ and $\phi(v)>0$, and since it is $\phi\in\mathcal{C}^0$, 
	then there must be an intermediate point $c'\in\overline{p'v}$ with $\phi(c')=0$. 
Finally, because $\phi(c) = \phi(c')=0$, it is 
\begin{eqnarray*}
 0&=&\frac{\|v-c\|}{\|c-c'\|} \left(\phi(c) - \phi(c')\right) = \e^{-2}  \frac{\|v-c\|}{\|c-c'\|}  (w-v)^t \left[ Q_{c'}(m-c') - Q_c(m-c)\right] \\
	&=& \e^{-2}  \frac{\|v-c\|}{\|c-c'\|} \left[ (c-c')^t Q_c (w-v) + (w-v)^t(Q_{c'}-Q_c)(m-c')\right] \\
	&=& \e^{-2} (c-v)^t Q_c (w-v) + \e^{-2}  \frac{\|v-c\|}{\|c-c'\|}  (w-v)^t (Q_{c'}-Q_c)(m-c')
\end{eqnarray*}

If we define $\alpha = |\e^{-2}  \frac{\|v-c\|}{\|c-c'\|}  (w-v)^t (Q_{c'}-Q_c)(m-c')|$, and 
	$\beta = |\e^{-2} (c-v)^t Q_c (w-v)|$, then
\[ \frac{\|v-c\|}{\|c-c'\|} |\phi(c) - \phi(c')| \ge \beta - \alpha \]
We can reach a contradiction by showing that $\phi(c) - \phi(c')$ does not vanish. 
To do this, it suffices to bound $\alpha$ from above, and $\beta$ from below in such a way that their difference is always positive. 
In particular, we will see that $\alpha$ can be made arbitrarily small by 
	requiring $V$ to be an $\e$-cover of sufficiently small $\e$. 
To bound $\beta$ from below, on the other hand, it is not sufficient for $V$ to form a sufficiently dense cover. 
$\beta$ is sensitive to both the density of sites in $V$, as well as their relative distribution. 
It is, however, possible to find a sufficiently-high lower bound of $\beta$ by requiring $V$ to be an asymmetric $\e$-net. 
The asymmetric $\e$-net condition is therefore sufficient to bound both $\beta$ from below, and $\alpha$ from above, in such a way as to ensure 
that $\phi(c)-\phi(c')$ doesn't vanish, creating a contradiction and concluding the proof. 
The following two lemmas provide the relevant bounds for $\alpha$ and $\beta$. 
Auxiliary lemmas from the Appendix are used in the proofs. 
\\


\begin{lem}\label{Vlm:alpha0}
	Given an asymmetric $\e$-cover $V$,  if $v,w\in V$ are Voronoi-neighbors, and it is $q\in R^{{ }^{DW}}_v$ and $c,c'$ and $m$ as described above, then 
		\[ \alpha = | \e^{-2}  \frac{\|v-c\|}{\|c-c'\|}  (w-v)^t (Q_{c'}-Q_c)(m-c')| \le 2\left(\e\s\right)^2 + 4\e\s \]
\end{lem}
\begin{proof}
	Since $c$ is in the Voronoi regions of $v,w$, it is $\|M_c(c-v)\| = \|M_c(c-w)\|\le\e$ and therefore $\|M_c(w-v)\| \le \|M_c(w-c)\| + \|M_c(c-v)\| \le 2\e$. 
	Likewise, it is straightforward to show that $\|M_c(c-c')\|\le\|M_c(c-v)\|\le\e$ implies $\|M_c(m-c')\|\le\e$. 
	Therefore, it is
	\begin{eqnarray*}	
		\alpha &=& | \e^{-2}  \frac{\|v-c\|}{\|c-c'\|}  (w-v)^t M_c^t \left[ M_c^{-t}(Q_{c'}-Q_c) M_c^{-1} \right] M_c (m-c')| \\
				&\le& | \e^{-2}  \frac{\|v-c\|}{\|c-c'\|}  \|M_c(w-v)\| \rho\left( M_c^{-t}(Q_{c'}-Q_c) M_c^{-1} \right) \|M_c (m-c')\| \\
				&\le& 2 \frac{\|v-c\|}{\|c-c'\|} \rho( M_c^{-t}(Q_{c'}-Q_c) M_c^{-1} )
	\end{eqnarray*}
	
	Letting $A=M_{c'}M_c^{-1}-I$, it is $M_c^{-t}(Q_{c'}-Q_c) M_c^{-1} = A^t A + A + A^t$, and so
		$\rho( M_c^{-t}(Q_{c'}-Q_c) M_c^{-1}) = \rho(A)^2 + 2\rho(A)$. 
	By Lemma~\ref{Vlm:cover0}, 
		it is 
	\[ \rho(A) = \rho\left((M_{c'} - M_c)M_c^{-1}\right) \le \s \|M_c(c-c')\| \]
	and, using $\|M_c(c-c')\|\le\|M_c(c-v)\|\le\e$, it is
	\begin{eqnarray*}	
		\frac{\|v-c\|}{\|c-c'\|} \rho( M_c^{-t}(Q_{c'}-Q_c) M_c^{-1} ) &\le& \frac{\|v-c\|}{\|c-c'\|} \left[ \rho(A)^2 + 2\rho(A) \right] \\
				&\le& \frac{\|v-c\|}{\|c-c'\|} \left[ \s^2 \|M_c(c-c')\|^2 + 2\s \|M_c(c-c')\| \right] \\
				&\le&  \left[ \s^2 \|M_c(c-v)\|^2 + 2\s \|M_c(c-v)\| \right] \\
				&\le&  \left(\e\s\right)^2 + 2\e\s
	\end{eqnarray*}
	which in turn implies $\alpha \le 2\left(\e\s\right)^2 + 4\e\s$.
\end{proof}
\vspace*{0.1in}

%

%

\begin{lem}\label{Vlm:beta0}
	Given an asymmetric $\e$-net $V$,  if $v,w\in V$ are Voronoi-neighbors, and it is $q\in R^{{ }^{DW}}_v$ and $c$ as described above, then it is
		\[ \beta  =   |\e^{-2} (c-v)^t Q_c (w-v)| \ge 1/\left( 2 (1+\e\s)^2 \right) \]
\end{lem}
\begin{proof}
	We first show that $(w-v)^t Q_c (v+w - 2c) = 0$. 	
	Given that $D_Q^{{ }^{DW}}(v,c) = D_Q^{{ }^{DW}}(w,c)$ implies $(v-c)^t Q_c (v-c) = (w-c)^t Q_c (w-c)$, it is
	\begin{eqnarray*}
		(w-v)^t Q_c (v+w - 2c) &=& \left[(w-c) + (c-v)\right]^t Q_c \left[(w-c) + (v-c)\right] \\ 
							&=& \left[(w-c)^t Q_c (w-c) - (v-c)^t Q_c (v-c)\right] \\
							&+& \left[ (w-c)^t Q_c (v-c) + (c-v)^t Q_c (w-c) \right] = 0
	\end{eqnarray*}
	By Lemma~\ref{lem:1} it is
	\begin{eqnarray*}
		\frac{\e^2}{(1+\e\s)^2} &\le& |(w-v)^t Q_c (v-w) |\\
				& =& |2(w-v) Q_c (c-v) + (w-v)^t Q_c (v+w - 2c)| \\
				&=& 2|(w-v)^t Q_c(c-v)|
	\end{eqnarray*}
	and thus $|\e^{-2} (c-v)^t Q_c (w-v)| \ge 1/\left( 2 (1+\e\s)^2 \right)$. 
\end{proof}
\vspace*{0.1in}

The next theorem uses the bounds of Lemmas~\ref{Vlm:alpha0} and~\ref{Vlm:beta0} to prove that, 
	under certain circumstances, the difference $\phi(c)-\phi(c')$ cannot vanish and therefore the anisotropic Voronoi diagram 
	is orphan-free. \\
	
\begin{thm}\label{Vth:noorphans0}
	Given a continuous metric $Q$, 
		the Du/Wang diagram of an asymmetric $\e$-net (with respect to $D_Q^{{ }^{DW}}$) is orphan free if $\e\s \le 0.09868$. 
\end{thm}
\begin{proof}
Given the construction at the beginning of this section, it must be $\phi(c)-\phi(c')=0$. However, if $\e\s \le 0.09868$, it is
\[  \frac{\|v-c\|}{\|c-c'\|} |\phi(c) - \phi(c')| \ge \beta - \alpha \ge 1/\left( 2 (1+\e\s)^2 \right) - 2\left(\e\s\right)^2 + 4\e\s > 0 \]
reaching a contradiction. 

Since all points in $\overline{qv}$ must be closer to $v$ than to any $w\ne v$, $q$ cannot be in an orphan region of $v$. 
Additionally, since for every point $q\in R^{{ }^{DW}}_v$, the segment $\overline{qv}$ is also 
	in $R^{{ }^{DW}}_v$, every Voronoi region is star-shaped with respect to its generating site. 
\end{proof}
\vspace*{0.1in}

Finally, note that connectedness of Voronoi regions is shown by proving the stronger condition of star-shapedeness, which suggests that the condition $\e\s \le 0.09868$
	may be conservative in some cases. 

\subsection{Orphan-free Labelle/Shewchuk diagrams}\label{Vsec:LSnoorphans0}


In~\cite{LS}, a condition is described for a two-dimensional LS diagram to be orphan-free. 
Although this condition is somewhat technical, an accompanying iterative-insertion algorithm is provided that, 
	given enough time, will output an orphan-free LS diagram. 
In this section we describe conditions for an LS diagram to be orphan-free in any number of dimensions. 
The conditions are very similar to those of Section~\ref{Vsec:noorphans0}, namely, that the set of generating sites form 
	an asymmetric $\e$-net with respect to $D_Q^{{ }^{LS}}$, with sufficiently small $\e$. 
The net requirement is somewhat natural in the sense that it implies that the sites are ``uniformly distributed" with respect to $D_Q^{{ }^{LS}}$.

Similarly as in Section~\ref{Vsec:noorphans0}, we consider a point $q$ that is in an orphan region of some $v\in V$.
Since $q$ is in an orphan region of $v$, the segment $\overline{qv}$ cannot be contained in $R^{{ }^{LS}}_v$, 
	and thus there must be $p'\in\overline{qv}$ that belongs to the Voronoi region of some $w\ne v$. 
This in turn implies that there are two distinct points $c\in\overline{qp'}$, $c'\in\overline{p'v}$ that are equidistant from $v,w$. 
If we  define the function
\[	\phi(p(\l)) = \e^{-2}\left[D_Q^{{ }^{LS}}(w,p(\l)) - D_Q^{{ }^{LS}}(v,p(\l))\right]  \]
with $p(\l) = v (1-\l) + c \l$, then it must be $\phi(c)=\phi(c')=0$. 


%
We now prove that $q$ cannot be in an orphan region by showing that $\phi(c)-\phi(c')$ cannot vanish, reaching a contradiction. 

It is
\begin{eqnarray*}
	0&=&\frac{\|v-c\|}{\|c-c'\|} |\phi(c) - \phi(c')| = \e^{-2}  \frac{\|v-c\|}{\|c-c'\|} | (c-c')^t Q_w (c'-w) - (c-c')^t Q_v (c'-v) | \\
					&=& | \e^{-2} (c-v)^t Q_w (v-w) + \e^{-2} (c-v)^t (Q_w - Q_v) (c'-v) |
\end{eqnarray*}
which, letting $\alpha = | \e^{-2} (c-v)^t (Q_w - Q_v) (c'-v)  |$, 
and
$ \beta(\l) = |\e^{-2} (c-v)^t Q_w (v-w)  | $, can be rewritten, 
similarly as in Section~\ref{Vsec:noorphans0}, as
\[ \frac{\|v-c\|}{\|c-c'\|} |\phi(c) - \phi(c')| \ge \beta - \alpha \]
We now prove upper, and lower bounds for $\alpha$ and $\beta$, respectively. 
\\

\begin{lem}
	Given an asymmetric $\e$-net $V$,  if $v,w\in V$ are Voronoi-neighbors, 
		and $q\in R^{{ }^{LS}}_v$, and $p(\l)$ is as described above, then it is
		\[ \alpha = | \e^{-2} (c-v)^t (Q_w - Q_v) (c'-v)  | \le \g^2 + 2\g \]
	where $\g = \e\s (1+k)$, and $k = (1+\e\s) / (1-\e\s)$. 
\end{lem}
\begin{proof}
	Since $\|M_v (c'-v)\| \le \|M_v(c-v)\|\le\e$, by Lemma~\ref{lem:4}, it is
	\begin{eqnarray*}
		\alpha&=& | \e^{-2}  (c-v)^t M_v^t M_v^{-t} (Q_w - Q_v) M_v^{-1} M_v (c'-v)  | \\ 
				&\le& \e^{-2} \|M_v(c-v)\| \|M_v(c'-v)\| \rho(M_v^{-t} (Q_w - Q_v) M_v^{-1}) \\
				&\le&  \e^{-2} \|M_v(c-v)\|^2 \rho(M_v^{-t} (Q_w - Q_v) M_v^{-1}) \\
				&\le& \left(\e\s (1+k)\right)^2 + 2 \e\s (1+k) = \g^2 + 2\g
	\end{eqnarray*}
\end{proof}
\vspace*{0.1in}

\begin{lem}
	Given an asymmetric $\e$-net $V$,  if $v,w\in V$ are Voronoi-neighbors, 
		$q\in R^{{ }^{LS}}_v$, and $p(\l)$ is as described at the beginning of Sec.~\ref{Vsec:LSnoorphans0}, then it is
		\[ \beta = | \e^{-2} (c-v)^t Q_w (v-w)  | \ge (k^2 - \g^2 - 2\g)/2 \]
	where $\g = \e\s (1+k)$, and $k = (1+\e\s) / (1-\e\s)$
\end{lem}
\begin{proof}
	Because $D_Q^{{ }^{LS}}(v,c) = D_Q^{{ }^{LS}}(w,c)$, by Lemma~\ref{lem:4}, it is
	\begin{eqnarray*}
		|(w-v)^t Q_w (v+w-2c) | &=& |(w-c)^t Q_w (w-c) - (v-c)^t Q_v (v-c) \\
							&-& (v-c)^t (Q_w-Q_v) (v-c)| \\
						&=& |(v-c)^t (Q_w-Q_v) (v-c)| \le \g^2+2\g
	\end{eqnarray*}

	By Lemma~\ref{lem:23} it is
	\begin{eqnarray*}
		\e^2/k^2 &\le& |(v-w)^t Q_w (v-w)| = |2 (w-v)^t Q_w (c-v) + (w-v)^t Q_w (v+w-2c)| \\
				&\le&  2|(w-v)^t Q_w(c-v)| + |(w-v)^t Q_w (v+w-2c)| \\
				&\le&  2|(w-v)^t Q_w(c-v)| + \g^2+2\g
	\end{eqnarray*}
	and therefore $\beta = | \e^{-2} (c-v)^t Q_w (v-w)  | \ge (k^2 - \g^2 - 2\g)/2$, as claimed. 
\end{proof}
\vspace*{0.1in}

These bounds on $\alpha$ and $\beta$ imply the following:\\

\begin{thm}\label{Vth:LSnoorphans0}
	Given a continuous metric $Q$, 
	the Labelle/Shewchuk diagram of an asymmetric $\e$-net (with respect to $D_Q^{{ }^{LS}}$) is orphan free if $\e\s \le 0.0584$. 
\end{thm}
\begin{proof}

Given the construction at the beginning of this section, it must be $\phi(c)-\phi(c')=0$. However, if $\e\s \le 0.0584$, letting $\g = \e\s (1+k)$, and $k = (1+\e\s) / (1-\e\s)$, it is
	\[ \frac{\|v-c\|}{\|c-c'\|} |\phi(c) - \phi(c')| \ge \beta - \alpha \ge (k^2 - \g^2 - 2\g)/2 - \g^2-2\g > 0 \]
reaching a contradiction. 

Since all points in $\overline{qv}$ must be closer to $v$ than to any $w\ne v$, $q$ cannot be in an orphan region of $v$. 
Additionally, since for every point $q\in R^{{ }^{LS}}_v$, the segment $\overline{qv}$ is also 
	in $R^{{ }^{LS}}_v$, every Voronoi region is star-shaped with respect to its generating site. 
\end{proof}
\vspace*{0.1in}

\section{Conclusion}

This paper presents a simple and natural condition for the two definitions of anisotropic diagrams of~\cite{LS} and~\cite{DW} to be 
	composed of connected (and in particular star-shaped) regions. 
The condition is simply that the generating sites be roughly ``uniformly" distributed (forming an asymmetric $\e$-net), with respect to the underlying metric. 
Apart from being natural, this condition is also commonly employed for certain practical problems. 
In particular, for optimal quantization, where we are interested in the primal Voronoi diagram, the optimal quantization sets form an $\epsilon$-net~\cite{Gruber,enets}. 
For $\mathcal{L}^\infty$ PL approximation of functions~\cite{enets}, where we are interested in the dual simplicial complex, 
	existing asymptotically-optimal constructions use vertex sets that form an $\epsilon$-net~\cite{enets}.

Note that, although the definition of $\epsilon$-net must be slightly modified for the problem that concerns us here, 
	this modification is not of great practical importance since existing algorithms for computing $\epsilon$-nets~\cite{Gonz} 
	are easily adapted to produce the desired (asymmetric) $\epsilon$-net. 
	
As mentioned in Sec.~\ref{Vsec:setup}, computing an $\e$-net using the algorithm of~\cite{Gonz} involves, at each iteration, 
	computing the farthest point from the current set of sites: a task equivalent in cost to computing the Voronoi diagram of each intermediate set of sites. 
This makes explicitly computing $\s$ unnecessary since computing an asymmetric $\e$-net of sufficiently small $\e$ has a similar cost to 
	iteratively running the algorithm of~\cite{Gonz} while checking whether intermediate diagrams are orphan-free, 
	and then stopping when the first orphan-free diagram is produced. 
As mentioned earlier, the proofs in this paper guarantee that such an iterative algorithm will stop, while the specific bounds may give some 
	indication of when this happens. 

In the eventuality that ways to compute (asymmetric) $\e$-nets arise that are more efficient, 
it may be that computing $\s$ becomes useful since, along with Theorems~\ref{Vth:noorphans0} and~\ref{Vth:LSnoorphans0},
	it would provide a simple lower bound on the largest $\e$ for which an asymmetric $\e$-net is guaranteed to result in an orphan-free diagram. 
In this case it would become useful to know of more efficient ways to compute $\s$. 
In particular, in~\cite{techreport} [this reference is to a supplementary document included in the submission, which will be cited as a Technical Report in the final version], 
	we show that if the metric $Q$ is continuously differentiable, 
	more efficient ways to bound $\s$ exist, which involve looking at every point in the domain only once, as opposed to comparing all pairs of points (as in Def.~\ref{def:sigma0}). 
	In particular, if the metric is specified as a PL function over a simplicial complex, then bounding $\s$ only requires computing a single number at each element (in constant time), 
		and then taking the maximum over all elements, and has therefore linear complexity in the number of elements in the complex.

%

Finally, note that, 
aside from the well-behaved-ness implied by the lack of orphan regions in the Voronoi diagrams, 
	it is possible that orphan-freeness may be useful in guaranteeing properties of their dual abstract simplicial complexes such 
	as being absent of inverted elements. 
In future work, we would like to explore  whether orphan-freeness in a Voronoi diagram can be used, possibly along with additional conditions, to guarantee that it's dual is an embedded simplicial complex.

\bibliographystyle{plain}
\bibliography{avd}

\section*{Appendix} 

Assume given $\e>0$ and a metric $Q\in\mathcal{C}^0$ defined over domain $\O$, and let $k\equiv (1+\e\s) / (1-\e\s) > 1$. 
The following lemmas are used in the proofs of Section~\ref{partI}.


\begin{lem}\label{lem:rhom}
	Given a non-singular matrix $A\in\mathbb{R}^{n\times n}$, it is $\rho(A^{-1}) = \rho_m(A)^{-1}$. 
\end{lem}
\begin{proof}
	If $\l_i$, $i=1,\dots,n$ are the eigenvalues of A then
		\[ \rho(A^{-1}) = \max_i |\l_i^{-1}| = \max_i |\l_i|^{-1} = (\min_i |\l_i|)^{-1} = \rho_m(A)^{-1} \]
\end{proof}
\vspace*{0.1in}

\begin{lem}\label{lem:1}
	Let $V$ be an asymmetric $\e$-net w.r.t. $D_Q^{{ }^{DW}}$, and $v,w\in V$ be Voronoi neighbors of the resulting DW-diagram. 
	If $c\in\O$ is in the Voronoi regions of $v,w$ then
		\[ \|M_c (v-w)\| > \e / (1+\e\s) \]
\end{lem}
\begin{proof}
	Since $V$ is an asymmetric $\e$-net, it is either $D_Q^{{ }^{DW}}(v,w) > \e$ or $D_Q^{{ }^{DW}}(w,v)>\e$. 
	Assume w.l.o.g.\  that $D_Q^{{ }^{DW}}(w,v) > \e$ and therefore $\|M_v(v-w)\| > \e$. 
	Since $c$ is in the Voronoi regions of $v,w$, it must be $D_Q^{{ }^{DW}}(v,c) = D_Q^{{ }^{DW}}(w,c)\le \e$. Therefore, by Lemmas~\ref{Vlm:cover0} and~\ref{lem:rhom}, it is
	\begin{eqnarray*}
		\|M_c(v-w)\| &=& \|M_c M_v^{-1}M_v (v-w)\| \ge \rho_m(M_c M_v^{-1}) \|M_v(v-w)\| \\
					&=& \rho(M_v M_c^{-1})^{-1} \|M_v(v-w)\| > \e / (1+\e\s)
	\end{eqnarray*}
\end{proof}
\vspace*{0.1in}

\begin{lem}\label{lem:v2w}
	Let $V$ be an asymmetric $\e$-net w.r.t. $D_Q^{{ }^{DW}}$ (resp. $D_Q^{{ }^{LS}}$), and $v,w\in V$ be Voronoi neighbors of the resulting DW diagram (resp. LS diagram). 
	Then 
		\[  1/k \le \rho_m(M_w M_v^{-1}) \le \rho(M_w M_v^{-1}) \le k \]
\end{lem}
\begin{proof}
	Since $v,w$ are Voronoi neighbors, there is a point $c\in\O$ that belongs to the Voronoi regions of both  $v$ and of $w$. 
	In an DW diagram, it is $\|M_c(c-v)\| = \|M_c(c-w)\|\le\e$, and therefore, by Lemmas~\ref{Vlm:cover0} and~\ref{lem:rhom}, it is
		\begin{eqnarray*}
			 \rho(M_w M_v^{-1}) &\le& \rho(M_w M_c^{-1}) \rho(M_c M_v^{-1}) \\
			 					&=&  \rho(M_w M_c^{-1}) \rho_m(M_v M_c^{-1})^{-1} \le \frac{1+\e\s}{1-\e\s} = k
		\end{eqnarray*}
	and
		\begin{eqnarray*}
			 \rho_m(M_w M_v^{-1}) &\ge& \rho_m(M_w M_c^{-1}) \rho_m(M_c M_v^{-1}) \\
			 					&=&  \rho_m(M_w M_c^{-1}) \rho(M_v M_c^{-1})^{-1} \ge \frac{1-\e\s}{1+\e\s} = 1/k
		\end{eqnarray*}
		
	In an LS diagram, it is $\|M_v(c-v)\| = \|M_w(c-w)\|\le\e$, and therefore, by Lemmas~\ref{Vlm:cover0} and~\ref{lem:rhom}, it is
		\begin{eqnarray*}
			 \rho(M_w M_v^{-1}) &\le& \rho(M_w M_c^{-1}) \rho(M_c M_v^{-1}) \\
			 					&=&  \rho_m(M_c M_w^{-1})^{-1} \rho(M_c M_v^{-1}) \le \frac{1+\e\s}{1-\e\s} = k
		\end{eqnarray*}
	and
		\begin{eqnarray*}
			 \rho_m(M_w M_v^{-1}) &\ge& \rho_m(M_w M_c^{-1}) \rho_m(M_c M_v^{-1}) \\
			 					&=&  \rho(M_c M_w^{-1})^{-1} \rho_m(M_c M_v^{-1}) \ge \frac{1-\e\s}{1+\e\s} = 1/k
		\end{eqnarray*}
\end{proof}
\vspace*{0.1in}

\begin{lem}\label{lem:23}
	Let $V$ be an asymmetric $\e$-net w.r.t. $D_Q^{{ }^{LS}}$, and $v,w\in V$ be Voronoi neighbors of the resulting LS diagram. Then 
		\[ \e/k \le \|M_v (v-w)\| \le \e(1+k) \]
\end{lem}
\begin{proof}
	Since $v,w$ are Voronoi neighbors, there is $c\in\O$ that is in the LS Voronoi regions of $v,w$, 
		and therefore satisfies $\|M_v(v-c)\| = \|M_w(w-c)\|\le\e$. Thus, by Lemmas~\ref{Vlm:cover0} and~\ref{lem:v2w}, it is
	\begin{eqnarray*}
			\|M_v(v-w)\| &\le& \|M_v(v-c)\| + \|M_v(c-w)\| =  \|M_v(v-c)\| + \|M_v M_w^{-1} M_w(c-w)\| \\
						&\le& \|M_v(v-c)\| + \rho(M_v M_w^{-1}) \|M_w(w-c)\| \le \e \left(1+\rho(M_v M_w^{-1})\right) \le \e(1+k)
	\end{eqnarray*}
	and
		\[ \|M_v(v-w)\| \ge \|M_v M_w^{-1} M_w (v-w)\| \ge \rho_m(M_v  M_w^{-1}) \e \ge \e/k \]
\end{proof}
\vspace*{0.1in}

\begin{lem}\label{lem:4}
	Let $V$ be an asymmetric $\e$-net w.r.t. $D_Q^{{ }^{LS}}$, and $v,w\in V$ be Voronoi neighbors of the resulting LS diagram. Then 
		\[ \rho(M_v^{-t} Q_w M_v^{-1} - I) \le \left(\e\s (1+k)\right)^2 + 2 \e\s (1+k) \]
\end{lem}
\begin{proof}
	Let $A=M_w M_v^{-1}-I$ and $B=M_v^{-t} Q_w M_v^{-1} - I$, where it is $B = A^t A + A + A^t$. 
	By the definition of $\s$ and Lemma~\ref{lem:23}, it is 
		\[ \rho(A) = \rho(M_w M_v^{-1} - I) \le \s \|M_v(v-w)\| \le \e\s (1+k) \]
	and therefore
		\[ \rho(M_v^{-t} Q_w M_v^{-1} - I) = \rho(A^t A + A + A^t) \le \rho(A)^2 + 2\rho(A) \le  \left(\e\s (1+k)\right)^2 + 2 \e\s (1+k) \]
\end{proof}
\vspace*{0.1in}

\end{document}